\documentclass[english,11pt]{elsarticle}
\usepackage{graphicx}
\usepackage{epstopdf}
\usepackage{ae,aecompl}

\usepackage[T1]{fontenc}
\usepackage[latin9]{inputenc}
\usepackage[letterpaper]{geometry}
\usepackage{setspace}
\usepackage{amsmath}
\usepackage{amsthm}
\usepackage{bm}       % bold math
\usepackage{bbm}
\usepackage{dcolumn}  % Align table columns on decimal pointFBA)
\usepackage{graphicx} % Include figure files
\usepackage{graphics}
\usepackage{psfrag}
\usepackage{pstricks}
\usepackage{pst-node}

\usepackage[pdftex,
  pdftitle={A variational principle for computing nonequilibrium fluxes and
            potentials in genome-scale biochemical networks},
  pdfauthor={R. M. T. Fleming, C. M. Maes, M. A. Saunders, Y. Ye, and B. O. Palsson},
  pdfsubject={A steady-state nonequilibrium network of biochemical reactions},
  pdfkeywords={constraint-based modeling, flux balance analysis,
  thermodynamics, metabolism, variational principle, convex
  optimization, entropy function},
  pdfpagemode=UseOutlines,
  pdfstartview=FitH,
  bookmarks,
  bookmarksopen,
  colorlinks,
  linkcolor=blue,
  citecolor=black,
  urlcolor=blue,
  hypertexnames=false
]{hyperref}

\makeatletter

\newcommand*{\R}{\mathbbm{R}}
\newcommand*{\grad}{\nabla}

\newcommand*{\dims}[2]{\R^{#1 \times #2}}
\newcommand*{\diag}{\mathrm{diag}}

\newcommand*{\range}{\mathcal{R}}

\newcommand{\maximize}{\mathop{\operator@font{maximize}}}
\newcommand{\minimize}{\mathop{\operator@font{minimize}}}

\newcommand*{\st}{\text{subject to }}
\renewcommand{\l}{\ell}

\newtheorem{lemma}{Lemma}
\newtheorem{proposition}{Proposition}
\newtheorem{theorem}{Theorem}
\newcounter{definition}
\newenvironment{definition}[1][]
 {\refstepcounter{definition}{\bf Definition~\arabic{definition}.~#1}}{}

\newcommand{\ecoli}{\emph{E. coli}}

\makeatother

\begin{document}
 
\title{A variational principle for computing nonequilibrium fluxes and
  potentials in genome-scale biochemical networks}

%\author{R. M. T. Fleming}
%\affiliation{Center for Systems Biology, University of Iceland}
%\email{ronan.mt.fleming@gmail.com}
% \author{C. M. Maes}
% \affiliation{Institute for Computational and Mathematical Engineering, Stanford University}
% \author{M. A. Saunders and Y. Ye}
% \affiliation{Department of Management Science and Engineering, Stanford University}
% \author{B. \O.\ Palsson}
% \affiliation{Department of Bioengineering, University of California, San Diego}

\author{R. M. T. Fleming
     \\ Center for Systems Biology,
        University of Iceland
     \\ Ph: +354 618 6245, Email: ronan.mt.fleming@gmail.com
     \\ Sturlugata 8, Reykjavik 101, Iceland.
     \\[8pt] C. M. Maes
     \\ Institute for Computational and Mathematical Engineering,
        Stanford University
     \\[8pt] M. A. Saunders and Y. Ye
     \\ Department of Management Science and Engineering,
        Stanford University
     \\[8pt] B. \O.\ Palsson
     \\ Department of Bioengineering,
        University of California, San Diego}

\date{\today}

\begin{abstract}
  We derive a convex optimization problem on a steady-state
  nonequilibrium network of biochemical reactions, with the property
  that energy conservation and the second law of thermodynamics both
  hold at the problem solution. This suggests a new variational
  principle for biochemical networks that can be implemented in a
  computationally tractable manner.  We derive the Lagrange dual
    of the optimization problem and use strong duality
    to demonstrate that a biochemical analogue of Tellegen's
    theorem holds at optimality. Each optimal flux is dependent on a
    free parameter that we relate to an elementary kinetic parameter
    when mass action kinetics is assumed.
      
  % \and provides a computationally tractable method for enforcing
  % energy conservation and the second law of thermodynamics, in
  % addition to steady state mass conservation.
  
  % The method may be used for predicting behavior in genome-scale
  % biochemical networks such as those used by systems biologists in
  % models of particular organisms.

  \smallskip

  Keywords: constraint-based modeling, flux balance analysis,
            thermodynamics, convex optimization, entropy function
\end{abstract}

\maketitle

\section{Introduction}

The biochemical system of any organism can be represented
mathematically by a network of chemicals (nodes) and reactions
(edges).  To analyze these networks at genome scale, systems
biologists often use a linear optimization technique called flux
balance analysis (FBA) \citep{Savinell1992c}. Flux balance requires
that the sum of fluxes into and out of each node in the network be
zero. This is equivalent to Kirchhoff's current law in an electrical
network. Recent work has sought to augment flux balance analysis with
Kirchhoff's loop law for energy conservation as well as the second law
of thermodynamics
\citep{beard2002eba, looplaw, yang2005aip, nagrath2007iea,
fleming2008stk, schellenberger2011elimination}.

The incorporation of thermodynamic constraints into genome-scale
models has produced models that are biologically more realistic and
reveal greater insight into the control mechanisms operating in these
complex biological systems \citep{beard2002eba, kummel2006prs,
  xu2008optimization, Garg2010, liebermeister2010modular}. See
\citep{soh2010network} for a recent broad review of the application of
thermodynamic constraints to biochemical networks. The second law of
thermodynamics may be applied to each reversible elementary reaction
by constraining net reaction flux to be in the direction of a negative
change in chemical potential \citep{burton1957energy}.  Constraints on
net reaction flux can be incorporated within flux balance analysis as
additional linear inequality constraints \citep{fleming2009qas}.
Software packages to quantitatively assign reaction directionality for
genome-scale models of metabolism are available
\citep{fleming2011bertalanffy,cobraV2}.

In contrast to the addition of constraints arising from the second law
of thermodynamics, the addition of energy conservation constraints has
been problematic because the resulting equations are nonlinear and/or
nonconvex. Previous attempts required computing the global solution of
a nonconvex continuous optimization problem
\citep{beard2002eba,nagrath2007iea,fleming2008stk}, solving an NP-hard
problem \citep{yang2005aip}, or solving a mixed integer linear program
\citep{henry2007tbm,schellenberger2011elimination}.
Mixed integer programs have unpredictable computational complexity.

% \new{In the iterative process of reconstructing a biochemical network
%   from experimental data \citep{thieleTests}, a draft model can arise
%   that is incompatible with any steady state flux or chemical
%   potential.  If one attempted to solve such a model with a nonconvex
%   optimization algorithm, it is likely that one could not
%   distinguish between model infeasibility and inability of the
%   algorithm to find a solution to a feasible model. Such intractable
%   situations can be avoided if the model is formulated as a convex
%   optimization problem for which solvers guarantee finding a solution
%   (if one exists) or give a certificate of infeasibility otherwise.}
%  % if one formulates the problem as an
%  % optimization problem where one is guaranteed to find a solution
%  % efficiently, if a solution exists.}

The purpose of this work is to show that Kirchhoff's loop law and the
second law of thermodynamics arise naturally from the optimality
conditions of a convex optimization problem with flux balance
constraints.  Furthermore, every set of reaction fluxes that satisfies
Kirchhoff's loop law and the second law of thermodynamics must be
optimal for some instance of this problem.  This suggests that there
is an underlying variational principle operating in biochemical
networks. Our convex optimization formulation leads to
 polynomial-time algorithms \citep{Ye1997} for computing steady state fluxes 
 that also satisfy energy conservation and the
second law of thermodynamics.

\section{Linear resistive networks}

Consider a simple electrical circuit consisting of current
sources, batteries, and resistors, as illustrated in Figure \ref{circuit}.
\begin{figure}[!hb]
\vspace{15pt}
\centering
\vspace*{-3ex}
\psfrag{x1}{\large$x_1$}
\psfrag{x2}{\large$x_2$}
\psfrag{x3}{\large$x_3$}
\psfrag{x4}{\large$x_4$}
\psfrag{b1}{\large$b_1$}
\psfrag{f1}{\large$f_1$}
\psfrag{y1}{\large$y_1$}
\psfrag{y2}{\large$y_2$}
\psfrag{y3}{\large$y_3$}
\psfrag{y4}{\large$y_4$}
\scalebox{.8}{\includegraphics{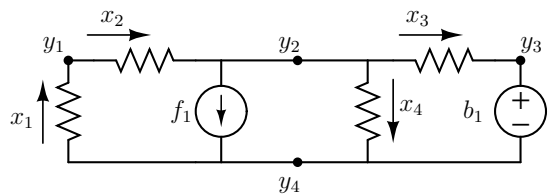}}
\vspace{-5pt}
%\caption{A linear resistive network with currents $x$, potentials $y$,
%  batteries $b$, and current sources $f$.}
\caption{A linear resistive network with currents $x$, potentials $y$,
  batteries $b$, and current sources $f$.}
\label{circuit}
\end{figure}
This is a linear resistive network with $m$ nodes and
$n$ edges, where the node variables $y \in \R^m$ represent
potentials and the edge variables $x \in \R^n$ represent flows (or
currents) in the network. The circuit topology is defined by a
node-edge incidence matrix $A \in \dims{m}{n}$, and properties of the
network are encoded in a set of data vectors: $f \in \R^m$ is a vector
of current sources, $b \in \R^n$ is a vector of batteries, and $r \in
\R^n$ is a vector of resistances ($r > 0$).

To solve for the voltages and currents in the circuit we use three
fundamental laws:
Kirchhoff's current law (KCL) $A x = f$,
Kirchhoff's voltage (or loop) law (KVL) $w = b + A^T y$,
and Ohm's law $w = Rx$ (where  $w \in \R^n$ is a vector of voltages
and $R = \diag(r)$ is a positive-definite diagonal matrix).
Maxwell's minimum heat theorem \citep{maxwell1873treatise} is a
variational principle that underlies this circuit. It seeks a set of
currents that minimize the heat (or power) dissipated subject to
KCL. This is the convex optimization problem \citep{strang1986introduction}
\begin{equation}
\begin{array}{ll}
   \minimize & F(x) \equiv \frac{1}{2} x^T R x - b^T x
\\ \st       & A x = f \quad : y
\end{array}
\tag{QP}
\label{linearresistive}
\end{equation}
where the node variables $y$ are Lagrange multipliers for the
equality constraints.  The optimality conditions %of this program
$\grad F(x) = A^T y$ yield equations that enforce KVL and Ohm's Law, and 
the optimal variables $x^\star$ and $y^\star$ are a set of consistent
potentials and currents for the circuit. 

Biochemical networks are significantly more complicated than linear
resistive networks.  However, some of the same underlying network
concepts apply \citep{oster1971nt}. In this work we construct an
optimization problem in a form similar to problem (QP), where the
potentials are Lagrange multipliers for an equality constraint on the
flow variables, and the optimality conditions of the problem yield
equations that enforce Kirchhoff's loop law and the second law of
thermodynamics. Previous work noted an ``analogy'' between Lagrange
multipliers and chemical potentials in flux balance analysis
\citep{warren2007dta}.  The key limitation of that work,
  concerned with duality in linear optimization, is that the
  optimality conditions of a linear optimization problem cannot
  enforce any relationship between net flux and change in chemical
  potential. In the present work, we establish a quantitative relation
  between the Lagrange multipliers and classical chemical potential
  for a nonlinear, yet convex, optimization problem in a framework
  consistent with classical thermodynamics.

\section{Biochemical networks}

The mathematical representation of a biochemical network is the
\emph{stoichiometric matrix} $S$.  Like $A$ above, $S \in \dims{m}{n}$
is a sparse incidence matrix that encodes the network
topology. However, biochemical networks, unlike linear resistive
networks, are nonlinear networks or hypergraphs. That is, a single
edge may link many nodes to many nodes, and the entries in $S$, which
are integer stoichiometric coefficients, are not confined to the set
$\{-1,0,1\}$. Each row of $S$ corresponds to an individual chemical
compound, and each column of $S$ corresponds to an individual
elementary reaction.  In practice, $m < n$ and $S$ does not have
full row-rank.  A model of a system is called genome-scale if a large
proportion of the system's genes are represented.  In current
genome-scale models of the metabolic system of \emph{E. coli}, $m$ and
$n$ are several thousand.

\emph{Flux balance analysis} (FBA) computes a set of fluxes that
satisfy steady state mass-conservation constraints and are optimal for
a biological objective function \citep{Savinell1992c,orth2010flux}.  A
\emph{flux} is a reaction rate; it represents flow through the network
and is analogous to current in an electrical circuit. We denote the
net flux of the $j$th reaction by the variable $v_j \in \R$. The
concentration of the $i$th chemical in the network is denoted $x_i \in
\R$.  A fundamental equation in flux balance analysis is the dynamic
mass conservation equation---a differential equation relating the
change in chemical concentration to reaction fluxes via the
stoichiometric matrix:
\begin{equation*}
   S v = \frac {dx}{dt}. 
\label{svequualdxdt}
\end{equation*}
Here $x \in \R^m$ is a vector of chemical concentrations and $v \in
\R^n$ is a vector of net reaction fluxes. Each row of this vector
equation states that the rate of change in concentration for a
chemical is the sum of fluxes that synthesize or degrade that
chemical. 

So far we have considered net flux for a stoichiometric matrix of
reactions, each of which conserves mass. A living biochemical system
operates in a \emph{nonequilibrium state} \citep{qian2007peh} and
exchanges mass with its surroundings.  This can be modeled by
including exchange reactions that do not conserve mass.  The model is
augmented with a matrix $S_e \in \dims{m}{k}$ and a corresponding set
of \emph{exchange fluxes} $v_e$, which are sources and sinks of
chemicals and are analogous to current sources in an electrical
network.  
In contrast to the columns of $S$, each column of $S_e$ 
corresponds to a reaction that does not conserve mass.
If we assume that the biochemical system is operating at a
steady state, then the concentrations of chemicals within the system
remain constant. Thus, we have
\begin{equation}
   S v + S_e v_e = \frac {dx}{dt} \equiv 0. 
\label{svequalzero}
\end{equation}
This equation is analogous to Kirchhoff's current law in an electrical network.
Henceforth we assume that there exists a feasible solution $(v,v_e)$ for
\eqref{svequalzero}.
Ensuring the existence of a steady state flux is part of a
quality control process during reconstruction of $S$ and $S_e$
from experimental literature \citep{thieleTests}.

Any reversible reaction $A + 2B \rightleftharpoons C$ may be split
into two one-way reactions: forward ($A + 2B \to C$) and reverse ($A +
2B \leftarrow C$). To distinguish between forward, reverse, and
exchange reactions we split the augmented stoichiometric matrix into
three components: $[S \,\ {-S} \,\ S_e]$.  Here $S$ contains all the
columns corresponding to forward reactions, and $-S$ contains all
columns corresponding to the reverse reactions.  The fluxes for these
one-way reactions form the vectors $v_f$ and $v_r$. The \emph{net
  flux} is then $v = v_f - v_r$. %\ref{fbafig} illustrates a biochemical network.

% \begin{figure}[!htbp]
% \centering
% \begin{pspicture}(0,2)(8,6.5)
% \rput(1,3){\circlenode{A}{A}}
% \rput(3,3){\circlenode{B}{B}}
% \rput(5,5){\circlenode{C}{C}}
% \rput(7,3){\circlenode{E}{D}}
% \pnode(0,3){F}
% \pnode(5,0){G}
% \pnode(5,6){H}
% \pnode(8,3){I}
% \psset{shortput=nab}
% \ncarc[arcangle= 15]{->}{A}{B}^{$v_{f1}$}
% \ncarc[arcangle=-15]{<-}{A}{B}_{$v_{r1}$}
% \ncarc[arcangle= 15]{->}{B}{C}
% \ncput*{$v_{f2}$}
% \ncarc[arcangle=-15]{<-}{B}{C}
% \ncput*{$v_{r2}$}
% \ncarc[arcangle= 15]{->}{C}{E}
% \ncput*{$v_{f3}$}
% \ncarc[arcangle=-15]{<-}{C}{E}
% \ncput*{$v_{r3}$}
% \ncarc[arcangle= 12]{->}{B}{E}
% \ncput*{$v_{f4}$}
% \ncarc[arcangle=-12]{<-}{B}{E}
% \ncput*{$v_{r4}$}
% \ncline{->}{E}{I}^[npos=.9]{$v_{e3}$}
% \ncline{->}{F}{A}^[npos=.0]{$v_{e1}$}
% \ncline{->}{H}{C}
% \nput[labelsep=0.5pt]{90}{H}{$v_{e2}$}
% \psframe[framearc=.4,linestyle=dashed,dash=3pt 2pt](0.4,2.3)(7.5,5.6)
% %\psgrid[subgriddiv=0]
% \end{pspicture}
% \caption{An example of a simple biochemical network. Chemicals are
%   represented as nodes, reactions as edges. The boundary of the network
%   is represented by a dashed line. Reaction fluxes are labeled as
%   exchange, forward, or reverse.}
% \label{fbafig}
% \end{figure}

Flux balance analysis has been implemented using linear programming \citep{Pal06}.
We take the flux balance analysis problem to be
\begin{equation}
\begin{array}{ll}
   \underset{v_f,v_r,v_e}{\maximize} & d^T v_e
\\ \st       & S v_f - S v_r +  S_e v_e = 0
\\           & v_f, v_r \ge 0, \quad \l \le v_e \le h.
\end{array}
\tag{FBA}
\label{fba}
\end{equation}
%% to compute a set of fluxes $v^\star = (v_f^\star,v_r^\star,
%% v_e^\star)$ that are optimal for a biological objective.
Often, the lower bounds $\l$ and upper bounds $h$ on the exchange
fluxes come from laboratory measurements (\emph{e.g.}\ the uptake of
glucose in a particular culture of \ecoli).  The vector $d$ is chosen
to optimize a biological objective (\emph{e.g.}\ maximizing
replication rate in unicellular organisms). Note that flux balance
analysis does not explicitly solve for $v_f$ and $v_r$ but rather $v =
v_f - v_r$.

We now prove a lemma about alternative optimal solutions to problem
(FBA).

\begin{lemma}\label{feaslemma}
If there is an optimal solution to problem (FBA), there is
an optimal solution with strictly positive internal fluxes $v_f$ and
$v_r$.
\end{lemma}

\begin{proof} %\small
   Let $(v_f^\star, v_r^\star, v_e^\star)$ be an optimal solution to
   problem (FBA). It could be the case that one or more components of
   $v_f^\star$ or $v_r^\star$ are zero; that is, they are sitting on
   their lower bounds. We want to show that we can construct a new
   optimal solution $(v_f^o,v_r^o,v_e^\star)$ with strictly positive
   fluxes $v_f^o$ and $v_r^o$. To do this, let $v_f^o = v_f^\star +
   \alpha e$ and $v_r^o = v_r^\star + \alpha e$, where $\alpha$ is any
   positive scalar and $e \in \R^n$ is the vector of all ones.  The
   internal fluxes $(v_f^o, v_r^o)$ are strictly positive, and the set
   of fluxes $(v_f^o, v_r^o, v_e^\star)$ is feasible because $\l \le
   v_e^\star \le h$ and
 \begin{align*}
    S v_f^o - S v_r^o + S_e v_e
         &= S (v_f^\star + \alpha e) - S (v_r^\star + \alpha e) + S_e v_e^\star
 \\      &= S v_f^\star - S v_r^\star + S_e v_e^\star = 0.
 \end{align*}
 Finally, the set of fluxes $(v_f^o, v_r^o, v_e^\star)$ is optimal
 because it has the same objective value $d^T v_e^\star$ as the
 optimal solution $(v_f^\star, v_r^\star, v_e^\star)$.  Here we used
 the fact that the objective function in problem (FBA) depends only on
 the exchange fluxes $v_e$; not on the internal fluxes $v_f$ and
 $v_r$.
\end{proof}

\section{Thermodynamic constraints}

Flux balance analysis predicts fluxes that satisfy steady state mass
conservation but not necessarily energy conservation or the second law
of thermodynamics. Whilst the domain of steady state mass conserved
fluxes includes those that are thermodynamically feasible, additional
constraints are required in order to guarantee a flux that
additionally satisfies energy conservation and the second law of
thermodynamics. Recent work has tried to add constraints based on
Kirchhoff's loop law and the second law of thermodynamics to problem
(FBA) \citep{beard2002eba,looplaw,nagrath2007iea}. However, these
constraints are problematic because they are nonlinear and
nonconvex. We now describe these constraints.

The loop law for chemical potentials in a biochemical network is
directly analogous to Kirchhoff's voltage law for electrical
circuits. It states that the stoichiometrically weighted sum of
chemical potentials around any closed loop of chemical reactions is
zero. We explicitly model a chemical potential $u_i \in \R$ for each
of the chemicals in the network, and we define the \emph{change in
  chemical potential} for all internal reactions in the network as the
vector
\begin{equation}
  \Delta u \equiv S^T u \quad (\in \R^n),
  \label{eq:Deltau}
\end{equation}
where $u \in \R^m$ is the vector of chemical potentials. This ensures
that Kirchhoff's loop law is satisfied, as energy conservation
requires that an injective relation exist between each row of $S$ and
a chemical potential \citep{planck1945tt}. It follows that the change
in chemical potential around a stoichiometrically balanced loop will
be zero.

Assuming mass-action kinetics, constant temperature and pressure, and
uniform spatial concentrations (\emph{i.e.}\ a well mixed system), it
is known (\emph{e.g.}\ \citep{ross2008taf}) that the change in
chemical potential may be expressed in terms of elementary one-way
reaction rates as
\begin{equation}
  \Delta u = \rho\log\left(v_r \,./\,v_f\right),
\label{logelementaryforwardreverse}
\end{equation}
where $./$ denotes component-wise division of vectors, and $\rho=RT >
0$ is the gas constant multiplied by temperature.  Equation
\eqref{logelementaryforwardreverse} leads directly to a macroscopic
(long-term) application of the second law of thermodynamics:
\[
   -\Delta u_j \, v_j = -\Delta u_j \, (v_{fj}-v_{rj}) \ge 0,
\]
which says that the net flux for each elementary reaction must be down
a gradient of chemical potential, that the system must dissipate heat,
and that entropy must increase as a result of work being done on the
system through the exchange fluxes \citep{qian2007peh}. The total heat
dissipation rate of the biochemical system in a non-equilibrium steady
state is given by $-\Delta u^T v \ge 0$.  Henceforth we consider each
flux as a dimensionless quantity. Dimensioned flux can be rendered
dimensionless by division with a standard flux of the same
dimensions. This standard may be defined according to a convenient
timescale in the same way that a standard concentration may be defined
according to a convenient abundance scale.

Under the specified assumptions, we now define a thermodynamically
feasible flux.

\smallskip

\begin{definition}
  For a network described by an augmented stoichiometric matrix $[S \
  \, {-S} \ \ S_e]$ with a given set of exchange fluxes $v_e$, a set
  of \emph{thermodynamically feasible fluxes} is a pair of internal
  flux vectors $(v_f,v_r) > 0$ that satisfy steady-state mass-balance,
\begin{equation}
  S v_f - S v_r = -S_e v_e, \label{massbalance}
\end{equation}
and for which there exists an underlying vector of chemical
potentials $u \in \R^m$ that satisfies \eqref{eq:Deltau}
and \eqref{logelementaryforwardreverse}:
\begin{equation}
\Delta u \equiv S^T u = \rho \log\left(v_r \,./\, v_f\right)
                        = \rho\log v_r - \rho\log v_f.
\label{fluxpotential}
\end{equation}
\end{definition}

\section{A variational principle}

We now present the main theorem of this paper. The theorem introduces
a new convex optimization problem with the same flux balance
constraints as problem (FBA) but with a negative entropy objective
function. It states that the thermodynamic constraints
\eqref{massbalance} and \eqref{fluxpotential} hold at its unique
solution.  We use $e$ to denote a vector of ones.

\begin{theorem}\label{optimalthmfeas}
  Let $v_e^\star$ be any set of optimal exchange fluxes from problem
  (FBA).  Define $b = -S_e v_e^\star$, and let $c$ be any vector in
  $\R^n$.  The convex equality-constrained problem
\begin{equation}
\begin{array}{ll}
   \underset{v_f,v_r > 0}{\minimize}
              & \phi \equiv v_f^T (\log(v_f) + c - e)
                          + v_r^T (\log(v_r) + c - e)
\\ \st        & S v_f - S v_r =  b \quad : y
\end{array}
\tag{EP}
\label{cvxeqp}
\end{equation}
is then feasible, and its solution $(v_f^\star, v_r^\star)$ is a set
of thermodynamically feasible internal fluxes.  The combined vector
$(v_f^\star,v_r^\star,v_e^\star)$ is thermodynamically feasible and
optimal for problem (FBA).  The associated chemical potentials $u$ may
be obtained from the optimal Lagrange multiplier $y^\star \in \R^m$
for the equality constraints according to $u = -2\rho y^\star$.
\end{theorem}

\begin{proof} %\small
  First note that the constraints $v_f, v_r > 0$ are implied by the
  domain of the logarithm; they have no associated Lagrange
  multipliers. From Lemma \ref{feaslemma}, we know that if $v_e^\star$
  is optimal for problem (FBA) and $b = -S_e v_e^\star$, there must be
  corresponding positive internal fluxes $v_f$ and $v_r$ that satisfy
  $S v_f - S v_r = b$. Therefore, problem (EP) with this choice of $b$
  is always feasible.

  Define the objective function as $\phi( v_f,v_r)$ and note that it
  is strictly convex because $\grad^2 \phi(v_f,v_r)$ is
  positive-definite for all $v_f,v_r > 0$, and it is bounded below.
  Problem (EP) is thus a convex linear equality-constrained problem
  with a unique optimal solution $(v_f^\star,v_r^\star$) that
  satisfies the optimality conditions
\begin{align}
   S^T y^\star &= \grad_{v_f}\phi = \log(v_f^\star) + c, \label{dualpos}
\\-S^T y^\star &= \grad_{v_r}\phi = \log(v_r^\star) + c, \label{dualneg}
\\ S v_f^\star - S v_r^\star &= b                       \label{primal}
\end{align}
for some vector $y^\star$ (which is not unique because $S$ has
low row rank).
Subtracting \eqref{dualpos} from \eqref{dualneg} gives $S^T(-2
y^\star) = \log(v_r^\star \,./\, v_f^\star)$. Taking $u = -2\rho
y^\star$ we see from \eqref{fluxpotential} that $(v_f^\star,
v_r^\star)$ is a pair of thermodynamically feasible fluxes with
underlying chemical potentials $u$ for the exchange fluxes
$v_e^\star$. The combined vector $(v_r^\star,v_f^\star,v_e^\star)$ is
feasible for problem (FBA), and is optimal because the objective $d^T
v_e^\star$ is unchanged.
\end{proof}

In summary, to compute an optimal solution
$(v_f^\star,v_r^\star,v_e^\star)$ to problem (FBA) that is
thermodynamically feasible, perform the following steps: solve problem
(FBA) to find an optimal exchange flux vector $v_e^\star$, form $b =
-S_e v_e^\star$, choose a vector $c$, and solve problem (EP) to find
$v_f^\star,v_r^\star > 0$.  Since $S$ is row rank deficient,
problem (EP) defines $\Delta u$ uniquely, but not the part of $u$ in
the nullspace of $S^T$. 
Observe that if $\phi( v_f,v_r)$ were a linear function of flux,
  one could not quantitatively relate flux to Lagrange multipliers at
  optimality, as primal and dual variables do not appear in the same
  optimality condition. This is the key limitation of previous efforts
  to draw an analogy between chemical potential and Lagrange
  multipliers to constraints in a linear optimization problem
  \citep{warren2007dta}.

Note the structural similarity of problems (QP) and (EP). In both
cases the primal variables are the flows in the networks, the
constraints impose Kirchhoff's current law, and the dual variables for
these constraints are the potentials in the network. There is a clear
physical interpretation of the objective function in problem (QP) and
a variational principle in operation. We believe there must also be a
variational principle in operation in biochemical networks. 
%However, it is not yet clear to the authors what this is, or why
%it appears in the form of the negative entropy objective.
In an abstract analysis of nonlinear resistive networks
  \citep{millar1951cxvi}, one would refer to the objective in (EP) as
  the \emph{resistive content}, as it is the sum of the areas under
  the \emph{constitutive relationships between flux and change in
  potential, for each reaction.}  From an information theoretic
  perspective, by
maximizing the entropy of the internal fluxes, problem (EP) gives
the most unbiased prediction \citep{Jaynes03} of
internal elementary fluxes, subject to mass-conservation constraints
and boundary conditions imposed by exchange fluxes.

The next theorem proves that if there is a set of thermodynamically
feasible fluxes in a biochemical network, it \emph{must} be the
solution to an optimization problem in the form of problem (EP).

\begin{theorem}\label{thmfeas}
  Every set of thermodynamically feasible fluxes $v_f$, $v_r$ (and
  chemical potentials $u$) is the solution (and corresponding Lagrange
  multiplier) of a convex optimization problem in the form of problem
  (EP).
\end{theorem}

\begin{proof} %\small
  We show how to choose the vector $c \in \R^n$ so that the given
  $v_f$, $v_r$, and $u$ are optimal. Since $v_f$ and $v_r$ satisfy
  \eqref{massbalance}, they are feasible. From
  \eqref{fluxpotential} we have $S^T u = \rho \log(v_r \,./\, v_f)$.
  Letting $y = -\frac{1}{2\rho} u$ we have $u = -2\rho y$ and hence
\[
  -2 S^T y = \log(v_r) - \log(v_f).
\]
If we take $2c =  -\log(v_r) - \log(v_f)$, this gives
\begin{align*}
    2 S^T y &= 2 \log(v_f)  + 2 c
\\ -2 S^T y &= 2 \log(v_r) + 2 c,
\end{align*}
which with \eqref{massbalance} are the optimality conditions for
problem (EP) as given by \eqref{dualpos}--\eqref{primal}. Therefore,
with $c = - \frac{1}{2} \log(v_r) - \frac{1}{2} \log(v_f)$, the given
$v_f$ and $v_r$ are the optimal solution, with Lagrange multiplier
$y=-\frac{1}{2\rho} u$.
\end{proof}

The vector $c$ in problem (EP) is a set of free parameters.  Theorem
\ref{optimalthmfeas} tells us that regardless of the value of $c$,
(EP) has a unique solution that is a thermodynamically
feasible flux. Theorem \ref{thmfeas} states that given any
thermodynamically feasible flux there exists at least one associated
vector $c$.

\section{Dual variational principle}

Given the primal optimization problem (EP) it is instructive to
consider the equivalent dual optimization problem, as its objective
function lends insight into the properties of the optimal dual
variables for the primal problem.

By Lemma~\ref{feaslemma}, the constraints of Problem (EP) have feasible
solutions that are strictly positive.  Similarly, there must be
feasible solutions that are finite, and these have finite objective
values.  Hence the optimum of the strictly convex objective function
of problem (EP) is finite and attainable.

\smallskip

\begin{proposition}
  The Lagrange dual of problem (EP) is the unconstrained
  convex optimization problem
\begin{equation}
  \maximize_y \ \psi(y) \equiv b^Ty -e^T\exp(S^Ty-c) - e^T\exp(-S^Ty-c).
  \tag{DEP}
  \label{eq:dualEP}
\end{equation}
\end{proposition}

\begin{proof} %\small
%That is,
%\begin{equation*}
%  \grad \psi(y^\star) \equiv b - S \exp(S^Ty^\star - c) - S\exp(-S^Ty^\star - c) %= 0.
%\end{equation*}
The dual for problem (EP) is derived in terms of the associated Lagrangian,
\begin{equation*}
  \mathcal{L}\left(v_{f},v_{r},y\right) =
    v_{f}^{T}(\log(v_{f}) + c - e) + v_{r}^{T}(\log(v_{r}) + c - e)
  - y^{T}(Sv_{f}-Sv_{r}-b).
\end{equation*}
The dual function $\psi(y)$ is the set of greatest lower bounds
(denoted $\inf$ for infimum) of the Lagrangian over the primal
variables. Equivalently, it is the set of least upper bounds (denoted
$\sup$ for supremum) of the negative Lagrangian:
\[
  \psi(y) = \inf_{v_f,v_r}\  \mathcal{L}\left(v_{f},v_{r},y\right) =
          - \sup_{v_f,v_r}\ -\mathcal{L}\left(v_{f},v_{r},y\right).
\]
The supremum of the linear function $y^T b$ is itself,
and the other terms may be grouped to give
\begin{align}
  \psi(y) &= b^{T}y
           - \sup_{v_f} \left( y^T S v_f - v_f^T(\log(v_f)+c-e) \right)
           - \sup_{v_r} \left(-y^T S v_r - v_r^T(\log(v_r)+c-e) \right).
   \label{eq:psi(y)_sup}
\end{align}
%Each supremum is a Legendre transform of a differentiable, separable
%part of the objective function in problem (EP), and has the form of a
%conjugate function $f^{\star}(q)\equiv\sup_{p}(q^{T} p-f(p))$.  For
%the first supremum, we have $p=v_f$, $q = S^T y$, and $f(v_f) = v_f^T
%(\log(v_f) + c - e)$.
The first supremum is attained when the partial
derivative with respect to $v_f$ is zero:
\begin{equation}
   %\nabla_{v_f} f^{\star}(v_f) = 
   S^T y - \log(v_f) - c = 0
   \qquad\Leftrightarrow\qquad
   v_f = \exp(S^T y - c).
\label{eq:primaldualmapping1}
\end{equation}
Similarly for the second supremum,
\begin{equation}
   %\nabla_{v_r} f^{\star}(v_r) =
   -S^T y - \log(v_r) - c = 0
   \qquad\Leftrightarrow\qquad
    v_r = \exp(-S^T y - c).
\label{eq:primaldualmapping2}
\end{equation}
Substituting
\eqref{eq:primaldualmapping1}--\eqref{eq:primaldualmapping2} into
\eqref{eq:psi(y)_sup} gives the Lagrange dual problem \eqref{eq:dualEP}.
\end{proof}

  The first and second derivatives of $\psi(y)$ are
\begin{align*}
   \grad   \psi(y) &= b - S \exp(S^Ty - c) + S\exp(-S^Ty - c),
\\ \grad^2 \psi(y) &= - S D S^T,
\end{align*}
where $D \equiv \diag(\exp( S^Ty - c) + \exp(-S^Ty - c))$
is a positive-definite diagonal matrix for all finite $y$
(so that $\grad^2 \psi(y)$ is negative semidefinite).
The necessary first- and second-order conditions for a point
$y = y^\star$ to be an optimum are that $\grad \psi(y^\star) = 0$
and $\grad^2 \psi(y^\star)$ be negative semidefinite.
With $v_f^\star \equiv \exp(S^Ty^\star - c)$ and
     $v_r^\star \equiv \exp(-S^Ty^\star- c)$
we see that $(v_f^\star,v_r^\star,y^\star)$ satisfy the optimality
conditions for both (EP) and (DEP).

The dual objective $\psi(y)$ offers a complementary insight into the
meaning of our variational principle. The chemical potential is
defined by $u = -2\rho y^\star$. Up to scalar multiplication, $-b^Ty =
\frac{1}{2\rho}b^Tu$ corresponds to the rate of work being done by the
environment to maintain the system away from equilibrium
\citep{QB05}. Therefore, one may interpret the Lagrange dual problem
as a minimization of this rate of work, balanced against minimization
of the sum of thermodyamically feasible forward and reverse fluxes.
We note that minimization of a weighted linear combination of forward
or reverse fluxes, where thermodynamic equilibrium constants are used
as weighting factors, has previously been explored as an optimality
principle for metabolic networks \citep{holzhutter2004principle}.

A comprehensive introduction to convex optimization is given in
\citep{boyd2004co}.  Problem (EP) is an example of a monotropic
optimization problem \citep{rockafeller1984nfamo}.  Further discussion
of such problems (and duality) can be found in
\citep{bertsekas1999nonlinear}.

\subsection{Strong duality and Tellegen's theorem}
\label{sec:Strong-duality-Tellegens}

Since problem (EP) is convex with a finite attainable optimum
objective and linear constraints, it is known that strong duality
holds \citep{bertsekas1999nonlinear}.  Among other things, strong
duality means that the values of the primal and dual objectives are
equal at the optimum: $\phi(v_f^\star,v_r^\star) = \psi(y^\star)$.
Omitting the $\null^\star$s and using
\eqref{eq:primaldualmapping1}--\eqref{eq:primaldualmapping2}, we have
\begin{align*}
   v_f^T(\log(v_f)+c-e) + v_r^T(\log(v_r)+c-e)
          &= b^Ty - e^T\exp(S^Ty-c) - e^T\exp(-S^Ty-c)
\\        &= b^Ty - e^T v_f - e^T v_r.
\end{align*}
Thus,
\begin{equation}
   v_f^T(\log(v_f)+c) + v_r^T(\log(v_r)+c) = b^Ty.  \label{eq:slaters3}
\end{equation}
Also from
\eqref{eq:primaldualmapping1}--\eqref{eq:primaldualmapping2} we have
\[
   c = -\log(v_f) + S^Ty = -\log(v_r) - S^Ty.
\]
With $u = -2\rho y$, \eqref{eq:slaters3} becomes
\begin{equation}
   u^TS(v_f-v_r) = u^Tb. \label{eq:tellegens}
\end{equation}
On the left of \eqref{eq:tellegens} is the rate of entropy production
by the biochemical network, and on the right is the rate of work done
by the environment to maintain the system away from equilibrium.
Their equality means that the rate at which heat is produced by
chemical reactions equals the rate at which heat is dissipated into
the environment, so the temperature of the system is time-invariant. In
the thermodynamic literature, \eqref{eq:tellegens} is known as the
isothermal Clausius equality \citep{QB05}, whereas from electrical
network literature one may recognize \eqref{eq:tellegens} as a
biochemical version of Tellegen's theorem. 

Oster and Desoer \citep{oster1971tellegen} previously recognized
  that a reaction flux vector satisfying Kirchhoff's current law and a
  chemical potential vector satisfying Kirchhoff's voltage law is
  necessary and sufficient for Tellegen's theorem to hold for a system
  of biochemical reactions. Although Tellegen's theorem holds
  irrespective of any constitutive relationship between flux and
  chemical potential, if one first assumes that unidirectional flux is
  a (strictly) monotonically increasing function of change in chemical
  potential, then one obtains flux and potential vectors satisfying
  Kirchhoff's laws with a single (strictly) convex optimization
  problem \citep{duffin1947nonlinear2a}.  The challenge is to pose an
  optimization problem with optimality conditions that enforce a
  constitutive relationship between flux and potential matching
  established theory in chemical kinetics.

\section{Thermodynamic feasibility and mass action kinetics}
\label{sec:Mass-action-kinetics}

From the optimality conditions for problem (EP), observe that forward
flux is a function of substrate and product chemical potentials. For
an elementary reaction, one would expect that forward flux should
depend only on substrate chemical potential(s) and reverse flux should
depend only on product chemical potential(s)
\citep{ederer2007tfk}. According to mass action kinetics, the rate of an elementary forward reaction only
depends on a forward kinetic parameter and substrate
concentration(s) \citep{Cor04}. In the objective of problem (EP), for simplicity of
exposition, consider replacing $c^{T}(v_{f}+v_{r})$ with
$c_{f}^{T}v_{f}+c_{r}^{T}v_{r}$. When $c_{f}-c_{r}$ is constrained to
lie in the range of $S^{T}$, this provides another mechanism for
satisfying \eqref{fluxpotential}. Here we show that one may choose
$c_{f},c_{r}$ such that mass action kinetics also holds.

The chemical potential of the $i$th metabolite in dilute solution at constant
temperature and pressure can be expressed as
\begin{equation}
   u_i = u_i^o + \rho\log\left(x_i/x_i^o\right),
   \label{eq:chemicalPotentialDefined}
\end{equation}
where $u_i^o$ is the standard chemical potential, $x_i$ is the
molar concentration of the metabolite, and $x_i^o$ is a reference
concentration that we define to be one molar. The standard chemical
potential is therefore the chemical potential of a metabolite in a
reference state. Assume we are given forward and reverse elementary
kinetic parameter vectors $k_f,k_r \in \R^n$ satisfying the
thermodynamic feasibility condition
\begin{equation}
   S^Tu^o = \rho \log (k_r \,./\, k_f)  \label{eq:haldane}
\end{equation}
for a particular standard chemical potential vector $u^o \in \R^m$.
Mass action kinetics may be expressed as a set of vectors
$v_f, v_r, k_f, k_r, x$ that satisfy
\begin{align*}
   \log(v_{f}) &= \log(k_f) + F^T\log(x),
\\ \log(v_{r}) &= \log(k_r) + R^T\log(x),
\end{align*}
where the stoichiometric matrix has been decomposed into the
entry-wise difference between forward and reverse stoichiometric
matrices, respectively $F,R\in \R_{\ge0}^{m,n}$, defined by
\[
\begin{array}{rcl@{\ }l}
   S_{ij}=0 & \rightarrow & \{F_{ij}=\;0,   & R_{ij}=0\},
\\ S_{ij}<0 & \rightarrow & \{F_{ij}=-S_{ij},& R_{ij}=0\},
\\ S_{ij}>0 & \rightarrow & \{F_{ij}=\;0,   & R_{ij}=S_{ij}\}.
\end{array}
\]
With respect to the forward reaction direction, each column of
\textbf{$F$} contains the absolute value of the stoichiometric
coefficient for each substrate, and the corresponding column of $R$
contains the stoichiometric coefficient for each product.

Given a consistent stoichiometric matrix and thermodynamically
feasible kinetic parameters, it is still an important open question
whether for all $b \in \range(S)$ there exists a metabolite
concentration vector $x$ satisfying steady-state mass action kinetics,
namely
\begin{equation}
   S \exp(\log(k_f) + F^T\log(x))
 - S \exp(\log(k_r) + R^T\log(x)) = b.
\label{eq:massActionKinetics}
\end{equation}
Ongoing work aims to establish the necessary and sufficient conditions
for this to be so \citep{HomoExistence}. Nonetheless, assuming
\eqref{eq:massActionKinetics} is satisfiable, there exist $c_{f}$ and
$c_{r}$ such that
\begin{equation*}
   c_f = -\log(k_f) + R^T y^\star, \qquad c_r = -\log(k_r) + F^T y^\star,
\end{equation*}
and optimality conditions
\eqref{eq:primaldualmapping1}--\eqref{eq:primaldualmapping2} become
\begin{equation*}
   \log(v_f^\star) = \log(k_f) - F^T y^\star, \qquad
   \log(v_r^\star) = \log(k_r) - R^T y^\star.
\end{equation*}
The optimal dual vector may then be equated with the negative of
logarithmic concentration $y^\star = -\log(x)$, leading to
satisfaction of mass action kinetics in a non-equilibrium steady
state.

\section{Discussion}

With problem (EP) and Theorem \ref{optimalthmfeas} we provide a new
variational principle that enforces steady state mass conservation,
energy conservation and the second law of thermodynamics for
genome-scale biochemical networks.  Moreover, we prove that all
thermodynamically feasible steady state fluxes are instances of
problem (EP) for a different free parameter vector $c\in \R^n$
(Theorem \ref{thmfeas}). The values of the free parameters influence
the optimal forward and reverse fluxes through the relation
$v_{f}^{\star}{\,.*\,}v_{r}^{\star}=\exp(-2c)$.  Varying $c$ may
change $v_{f}^{\star}$ and $v_{r}^{\star}$ but not the fact that the
corresponding solution is thermodynamically feasible and optimal for
flux balance analysis.  

To compute a thermodynamically feasible flux we must first choose a
value for the vector $c$.  Thus, there is some freedom left in the
model. We see this property of the model as an advantage rather than a
disadvantage, as thermodynamic feasibility is necessary but not
sufficient for satisfaction of mass action kinetics \citep{Jamshidi2008a}.  In the objective
of problem (EP), one could replace $c^{T}(v_{f}+v_{r})$ by
$c_{f}^{T}v_{f}+c_{r}^{T}v_{r}$.  Assuming the existence of a mass
action kinetic non-equilibrium steady state,
section~\ref{sec:Mass-action-kinetics} demonstrates that this state
corresponds to a particular choice of $c_{f},c_{r}$ given particular
kinetic parameters $k_{f},k_{r}$. Kinetic parameters evolve, in the
biological sense, subject to thermodynamic feasibility
\eqref{eq:haldane}.  Even if $c_{f},c_{r}$ could be chosen such that
mass action kinetics holds, the actual $k_{f},k_{r}$ that pertain to a
particular organism's biochemical network would remain unknown. This
problem is not particular to our modeling approach---it reflects a
paucity of suitable enzyme kinetic data in biochemistry generally
\citep{cornishbowden2005eck}. 

In addition to the constraints in problem (FBA), flux balance analysis
often includes extra inequalities on the net flux of internal
reactions.  It is not possible to incorporate explicit inequality
constraints on net flux into problem (EP) because if any such
inequalities were active at the optimum, the corresponding nonzero
dual variables would appear in a modification of the optimality
conditions \eqref{eq:primaldualmapping1}--\eqref{eq:primaldualmapping2} and
interfere with satisfaction of the thermodynamic constraints
\eqref{eq:Deltau} that we know must hold. In practical applications of
problem (EP) to genome-scale biochemical networks, for arbitrarily
chosen $c$, the omission of explicit bounds on net flux leads to a
subset of net fluxes with directions opposite to that known
biochemically.

At this point, one might think that an obvious application would be to
search for free parameters $c_f$, $c_r$ that minimize the Euclidean
distance between predicted and experimentally determined change in
chemical potentials or fluxes. In \emph{E. coli} metabolism
\citep{fleming2009qas}, for example, one can experimentally quantify
transformed reaction Gibbs energy (\emph{in vivo} change in chemical
potential) using quantitative measurement of absolute concentrations
\citep{bennet2009amc}, combined with experimentally derived
\citep{alberty2003tbr,alberty2006bta} or group contribution estimates
\citep{henry2006gst,henry2007tbm,jankowski2008,finley2009thermodynamic}
of standard transformed reaction Gibbs energy. However, obtaining
\emph{in vivo} chemical potentials is not a limitation to validation
of the practical utility of problem (EP). Even if one could find free
parameters $c_f$, $c_r$ corresponding to predicted chemical potentials
close to experimental data, the corresponding fluxes might not satisfy
mass action kinetics, which is known to hold for elementary chemical
reactions.

A reliable algorithm for satisfaction of mass action kinetics at a
non-equilibrium steady state is the subject of ongoing work
\citep{HomoExistence}. Given a consistent stoichiometric matrix $S$,
for there to exist a non-equilibrium steady state it is necessary that
the boundary condition be in the range of the stoichiometric matrix:
$b \in \range(S)$. Given thermodynamically feasible kinetic
parameters in addition, it is an important open problem to establish
if $b \in \range(S)$ is also sufficient for there to exist a
non-equilibrium steady state satisfying mass action kinetics.  It may
be that the set of feasible boundary conditions is a (proper) subset
of the range of $S$. In fact, the boundary condition obtained from
flux balance analysis $b=-S_ev_e^\star$ is an underdetermined
parameter for problem (EP) because $v_e^\star$ may not be unique,
given the absence of strict convexity in problem (FBA). 
%\new{We have provided an example where we choose to obtain $b = -S_e v_e^\star$ from
% Problem (FBA) but Problem (EP) is well defined for any choice of $b \in \range(S)$.}
Problem (EP) is well defined for any choice of linear objective
function $d^T v_e$ in problem (FBA). Different values of $d$ simply
yield different vectors $b = -S_e v_e^\star$. Problem (EP) may also be
used with other non-FBA methods that calculate
optimal exchange fluxes $v_e^\star$ provided these
methods yield a $b \in \range(S)$.

We consider problem (EP) to represent a mapping between a set of
parameters and the corresponding set of thermodynamically feasible
fluxes and potentials. More precisely, when $S$ is row reduced (to
have full row rank), problem (EP) represents a single-valued
surjective saddle point mapping \citep{rockafellar1997variational}
between a parameter vector $(c,b)$ and the corresponding primal-dual
optimal vector $(v_{f}^{\star},v_{r}^{\star},y^{\star})$. That is,
each parameter vector corresponds to a unique primal-dual optimal
vector, and every optimal vector is associated with at least one
parameter $(c,b)$.  Moreover, under the same conditions, this saddle
point mapping is locally Lipschitz continuous
\citep{dontchev2001primal}. That is, the optimal vector of fluxes and
potentials is smoothly perturbed by a smooth perturbation of the
parameter vector.

The properties discussed in the previous paragraph motivate future
efforts to design algorithms for gradient-based optimization of the
parameter vector. We envisage a convergent sequence of parametric
convex optimization problems whose final optimum satisfies mass-action
kinetic constraints. However, the necessary and sufficient conditions
for convergence of a sequence of parametric convex optimization
problems is still an active research area within convex
analysis. Exploitation of the properties of problem (EP) may provide
the route to such an algorithm.

The computational tractability of \emph{linear} flux balance analysis
(problem (FBA)) is one of the key reasons for its widespread use as a
genome-scale modeling tool. The theorems of this paper provide the
theoretical basis for efficiently computing thermodynamically feasible
fluxes, even for the largest genome-scale biological networks
\citep{thiele2009gcr}.  In particular:
\begin{itemize}
\item Problem (EP) is convex and includes only linear constraints.
  It is feasible whenever problem (FBA) is feasible,
  and its optimal solution is then unique.
\item Efficient polynomial-time algorithms
  \citep{andersen1998computational} exist for solving convex
  optimization problems of this form, based on interior methods
  \citep{Ye1997}.
\item For both (FBA) and (EP), these algorithms are guaranteed to
  return an optimal solution, or a certificate that the
  original problem is infeasible.
\end{itemize}
\noindent
In practice, computing thermodynamically feasible fluxes using
\emph{convex} flux balance analysis described herein should take no
longer than performing \emph{linear} flux balance analysis.

Although the ultimate value of this work lies in what is accomplished
with it in future biological studies, we believe it makes an important
step forward by providing the first computationally tractable method
for implementing energy conservation and the second law of thermodynamics
for genome-scale biochemical networks in a non-equilibrium steady
state. We also establish, in an exact manner, the duality relationship
between reaction rates and chemical potentials. Furthermore, our theorems
extend to any potential network where flux balance holds and the change
in potential can be written as a difference in a monotone function
of fluxes: $\Delta u = g(v_{r})-g(v_{f})$, where $g(\cdot):\R^{n}\to\R^{m}$
is monotone \citep{rockafeller1984nfamo}. These potential networks
admit a convex optimization model whose solution is unique and whose
potentials are Lagrange multipliers for the flux balance constraint.

\section*{Acknowledgements}

We would like to acknowledge invaluable assistance from Ines Thiele
and some helpful thermodynamic discussions with Hong Qian. 
We would also like to thank two anonymous reviewers
     for their constructive comments and for directing
     us to the 1951 paper by W. Millar \citep{millar1951cxvi}.
This project was supported by the U.S. Department of Energy (Office of
Advanced Scientific Computing Research and Office of Biological and
Environmental Research) as part of the Scientific Discovery Through
Advanced Computing program, grant DE-SC0002009.

\footnotesize
\frenchspacing

\bibliographystyle{jtb} % used this for the first jtb submission
%\bibliographystyle{jtbnew} 

%%% NEXT LINE IS FOR RONAN:
%\bibliography{/usr/local/texlive/2009/texmf-dist/bibtex/bib/references/Ronans_Biblio,/usr/local/texlive/2009/texmf-dist/bibtex/bib/references/thiele}

%%% MICHAEL AND CHRIS USE THIS:
%\input{aVariational2.bbl}

%\bibliography{/usr/local/texlive/2009/texmf-dist/doc/latex/ametsoc/bibliography/reference}
%\bibliography{reference}
%\bibliography{Ronans_Biblio}

\end{document}